\documentclass[a4paper,12pt]{amsart}
\usepackage[]{amsmath}
\usepackage{amsfonts}
\usepackage{amssymb}
\usepackage{xcolor}
\usepackage{mathrsfs}
\usepackage[numbers,sort]{natbib}
\usepackage{natbib}
\usepackage{enumerate}
\usepackage[normalem]{ulem}

\usepackage{tikz}
\usetikzlibrary{shapes.geometric,positioning,calc,cd}
\usepackage{binarytree} 
\usepackage[margin=3cm]{geometry} 
\usepackage{hyperref}

\usepackage{algorithm}
\usepackage{algpseudocode}

\newtheorem{thm}{Theorem}[section]
\newtheorem{prop}[thm]{Proposition}

\theoremstyle{definition}

\usepackage{amsfonts}
\usepackage{graphicx}

\newcommand{\PP}{\mathbb P}
\newcommand{\EE}{\mathbb E}
\newcommand{\RR}{\mathbb R}
\newcommand{\VV}{\mathbb V}
\newcommand{\FF}{\mathbb F}

\title[Biodiversity indices incorporating extinction risk]{Properties of biodiversity indices that incorporate future extinction risk}

\author{Mike Steel, Kristina Wicke and Arne Mooers}
\address{MS:  Biomathematics Research Centre, University of Canterbury, Christchurch, New Zealand}

\address{KW: Department of Mathematical Sciences, New Jersey Institute of Technology, Newark, USA \& National Institute for Theory and Mathematics in Biology, Northwestern University and The University of Chicago, Chicago, IL, USA}

\address{AM: Department of Biological Sciences, Simon Fraser University, Vancouver, Canada}
\date{\today}

\begin{document}
\begin{abstract}
The loss of biodiversity due to the likely widespread extinction of species in the near future is a focus of current concern in conservation biology.  One approach to measure the impact of this  extinction is based on the predicted loss of phylogenetic diversity. These predictions have  become a focus of the Zoological Society of London's `EDGE2' program for quantifying biodiversity loss and involves considering the HED (heightened evolutionary distinctiveness) and HEDGE (heightened evolutionary distinctiveness and globally endangered) indices which are based on phylogenetic diversity on a tree. Here, we show how to generalise the HED(GE) indices by expanding their application to more general settings (to phylogenetic networks, to feature diversity on discrete traits, and to arbitrary biodiversity measures).  We  provide a simple and explicit description of the mean and, importantly, the variance of such measures, and illustrate our results by an application to the phylogeny and a small set of features for all 27 extant Crocodilians.
\end{abstract}

\maketitle

Following Daniel P. Faith's seminal paper in 1992~(\cite{fai92}), phylogenetic diversity (PD) has become an increasingly prominent measure of biodiversity. Defined on a weighted phylogenetic tree, PD sums the branch lengths connecting a given subset of species, aiming to capture biodiversity more completely than species richness alone.

Preserving PD and, more broadly, the `Tree of Life' has become a central~(\cite{dia20, ver17}) though sometimes contested (\cite{guo, tuc19}) goal in conservation. Moreover, the concept of PD has inspired an extensive framework of PD-based metrics, measures that assign to each species a single numerical value representing its relative contribution to overall PD. These \emph{PD indices} provide a natural way to rank species according to their importance for preserving evolutionary history, and thus offer a simple basis for prioritisation decisions. Early examples include simple indices such as the \emph{fair proportion index} (\cite{red03,isa07}), also known as the \emph{evolutionary distinctiveness} score, which has been widely discussed in the literature and adopted in initiatives such as the initial ``Edge of Existence" (EDGE) program~(\cite{gum23, isa07}, see also \cite{red06}).

However, such simple indices have drawbacks: they typically do not account for variation in species’ extinction risks and can be sensitive to changes in the phylogeny, for instance following the extinction of other species (e.g., \cite{fis23,man24})). To address these limitations, extinction-aware indices incorporate the probability of species loss, thereby reflecting the expected contribution of each species to future PD. From this perspective, a species whose close relatives face high extinction risk carries greater responsibility for preserving shared evolutionary history than a species whose relatives are relatively secure. This idea underlies metrics such as the \emph{heightened evolutionary distinctiveness} (HED) and \emph{heightened evolutionary distinctiveness and globally endangered} (HEDGE) scores (\cite{ste07}), and has been incorporated in the updated EDGE2 framework (\cite{gum23}).
These metrics have been applied to quantify the potential future biodiversity loss under predicted extinction scenarios (e.g. the recent paper of \cite{for26} deals with the impact on the  flowering plant tree of life). 
Importantly, both extinction-informed PD and the EDGE metric have been adopted as indicators in the Kunming-Montreal Global Biodiversity Framework under multiple goals and targets (\cite{rob23}). This means they may be used both globally and nationally to report on progress towards official 2030 and 2050 goals and there exists an IUCN Species Survival Commission Task Force to aid in this work (\cite{gum22}).

In this paper, we generalise the HED and HEDGE scores of \cite{ste07} used in the EDGE2 approach, as well as the network-based extensions of PD introduced in \cite{wik18}.
Our extensions are in two directions. First, we  move beyond PD on trees and networks to more general diversity measures for assigning a score to any subset of species.  Some of these measures (e.g. feature diversity) share the properties that PD has of being submodular and monotone. However, these two properties could be overly restrictive.

For example, submodularity implies that the biodiversity score of a set $S$ of species (denoted $\varphi(S)$) is less than or equal to the sum of the scores of the individual species in $S$ (formally, $\varphi(S) \leq \sum_{x \in S} \varphi(x)$).  However, if the biodiversity score of $S$ incorporates positive ecological interactions (whereby `the whole is more than the sum of the parts') then this submodularity condition may well be inappropriate.

Similarly, monotonicity implies that the biodiversity score assigned by $\varphi$ to any subset $S$ of species never decreases if another species $x$ is added to $S$. However, if $\varphi$ is a measure of {\em future} biodiversity, and species $x$ is likely to cause the extinction of species in $S$, then monotonicity could also be problematic.

A second extension of the EDGE2 approach is to establish that the same type of product expression that arises for the EDGE2 score (an expected value of a random variable) can be applied to any biodiversity measure, generalising beyond phylogenetic diversity scores. We also provide an exact formula for the variance of the species-specific biodiversity metric underlying this generalisation of the EDGE2 score (Proposition~\ref{pro1}(ii)) and, in the case of feature diversity (of which PD is a special case), an explicit method for computing this variance (Proposition~\ref{pro2}).  This allows one to quantify the uncertainty associated with these scores, which is important when comparing species or prioritising conservation actions under stochastic extinction scenarios.  We apply our results to study the two EDGE2 indices (irreplaceability and expected gain following conservation) --  alongside their standard deviations -- in both the pure phylogenetic diversity (PD) and the more general feature diversity (FD) context for a small dataset of extant crocodiles. 

We begin by introducing some notation and recalling some earlier definitions.

\bigskip
\subsection{Definitions}
\label{def}

Let $X$ be any finite set whose elements we refer to as `species'. 
For each $x \in X$, each subset $S_x \subseteq X\setminus \{x\}$, each choice $I_x\in \{\emptyset, \{x\}\}$, and any function
$\varphi: 2^{X} \rightarrow \RR$, let
$$\Delta_\varphi(S_x,x) :=\varphi(S_x\cup\{x\}) - \varphi(S_x)$$
and
$$\Delta'_\varphi(S_x,x) :=\varphi(S_x\cup\{x\}) - \varphi(S_x \cup I_x).$$

When $\varphi$ is a biodiversity measure, $\Delta_\varphi(S_x,x)$ can be interpreted as the marginal gain in the diversity of $S_x$ obtained by actively conserving species $x$ versus its sure extinction. In contrast, $\Delta'_\varphi(S_x,x)$ measures the change in diversity of $S_x$ relative to a scenario in which $x$ is actively conserved and a scenario in which $x$ is not actively conserved and, therefore, may either go extinct ($I_x = \emptyset$) or remain extant ($I_x = \{x\}$).

Now consider these two indices when $S_x$ and $I_x$ are random variables. Specifically, suppose that each species $x \in X$ independently becomes extinct with probability $\epsilon_x$ (or not, with probability $1-\epsilon_x$). This model has been referred to as the  (generalised) field-of-bullets model of extinction, dating back to \cite{rau92}. We refer to it here as the {\em g-FOB model}. 
Thus, $I_x =\{x\}$ with probability $1-\epsilon_x$, and $I_x=\emptyset$ with probability $\epsilon_x$, and each other species $x' \in X\setminus \{x\}$ is (independently) present in the surviving set $S_x$ with probability $1-\epsilon_{x'}$.

For $x \in X$, let:
$$\Psi_x= \EE[\Delta_\varphi(S_x, x)] \mbox{ and } \Psi'_x = \EE[\Delta'_\varphi(S_x, x)].$$
When $\varphi$ is phylogenetic diversity on a tree, then $\Psi_x$ corresponds to the \emph{heightened evolutionary distinctiveness} (HED) score of species $x$, whereas $\Psi_x'$ is the \emph{heightened evolutionary distinctiveness and globally endangered} (HEDGE) score. By allowing $\varphi$ to be a more general function our approach here thereby generalises both these indices.

Both HED and HEDGE were originally defined in the context of PD by \cite{ste07} and correspond to the ED2 (irreplaceability) and EDGE2 (expected gain) metrics presented in \cite{gum23}.

In addition, we also consider the variance of these measures. Let
$$\Phi_x = \VV[\Delta_\varphi(S_x, x)]$$ and $$ \Phi'_x = \VV[\Delta'_\varphi(S_x, x)].$$

When $\varphi$ is some biodiversity measure applicable to any collection of species, 
the indices $\Psi_x$ and $\Psi_x'$ have a clear meaning. For example, 
$\Psi'_x$ is the expected additional biodiversity that results from protecting species $x$ (i.e. setting the extinction risk of $x$ to zero rather than leaving it at $\epsilon_x$) under a g-FOB model,  whereas $\Psi_x$ is the expected additional biodiversity (under the g-FOB model) that results from protecting species $x$ versus losing species $x$.

The importance of computing the variance values ($\Phi_x$ and $\Phi'_x$) is that they provide a measure of the extent to which the random variables $\Delta_\varphi(S_x,x)$ and $\Delta_\varphi'(S_x, x)$ are likely to deviate from their expected values $\Psi_x$ and $\Psi'_x$.  We will see below in our empirical (Crocodilian) data set that the standard deviation values exhibit wide variation across the species.

Notice also that although these indices associate a score to each single species, they take into account not only the extinction probability of this species but also of  other species in the data set.
Nonetheless, prioritising species according to their $\Psi$ and $\Psi'$ scores does not necessarily maximise summed diversity, however measured. In particular the top-$k$ species ranked by $\Psi$ or $\Psi'$ need not form a subset of $k$ species with the highest overall $\varphi$ score (see \cite{bor24, ste18,fai08} for relevant discussion). However, directly identifying subsets of species that maximise overall biodiversity presents its own challenges, including computational hardness, particularly in settings beyond simple PD.

\section{Expectations and variances of species-specific biodiversity metrics}

 Part (i) of the following result generalises a result from \cite{ste07} to allow the extension from the original setting of PD on trees to  any biodiversity measure (including, for example, PD on phylogenetic networks or, more generally, to feature diversity).  Part (ii) provides a measure of how well the generalised HEDGE score (i.e. the expected value  $\Psi'_x$ of $\Delta'_\varphi(S_x, x))$  approximates
$\Delta'_\varphi(S_x, x)$.
\label{exp}

\begin{prop}
\label{pro1}
\mbox{}
For \underline{any} function $\varphi:2^X \rightarrow \RR$, the following equalities hold under the g-FOB model:
\begin{itemize}
    \item[(i)] $\Psi'_x = \Psi_x \cdot \epsilon_x$.
    \item[(ii)]
    $\Phi'_x = \Phi_x\cdot \epsilon_x +\Psi_x^2\cdot \epsilon_x \cdot (1-\epsilon_x)$.
\end{itemize}
\end{prop}
Notice that when $\epsilon_x=0$ we have $\Phi'_x=0$, and when $\epsilon_x=1$ we have $\Phi'_x = \Phi_x$, as expected.
We now present the short proof of Proposition~\ref{pro1}.

{\em Proof.}
  To ease notation, let $\Delta = \Delta_\varphi(S_x, x)$, let $\Delta':= \Delta'_\varphi(S_x,x)$ (as defined earlier) and let $I=I_x$.
  
{\em Part (i)} 
$\Psi'_x=\EE[\Delta'] = \EE[\EE[\Delta'|I]]$ and 
\begin{equation}
\label{helps}
    \EE[\Delta'|I] = \begin{cases}
    0, & \mbox{if $I=\{x\}$ (w.p. $1-\epsilon_x$),}\\
    \Delta,  & \mbox{if $I = \emptyset$ (w.p. $\epsilon_x$).}
\end{cases}
\end{equation}
Thus, $\Psi'_x= 0 \cdot (1-\epsilon_x) + \EE[\Delta] \cdot \epsilon_x = \Psi_x \cdot \epsilon_x.$

{\em Part (ii)}
We have:
\begin{equation}
    \label{helps2}
    \Phi'_x = \VV[\Delta'] = \EE[(\Delta')^2] -\EE[\Delta']^2.
\end{equation}
Moreover, 
  $$\EE[(\Delta')^2] = \EE[\EE[(\Delta')^2|I]]=\EE[\Delta^2]\cdot \epsilon_x= (\Phi_x +\Psi^2_x)\cdot \epsilon_x,$$
  where the second equation follows by observing that when $I_x=\{x\}$, we have $\Delta'=0$ and when $I_x=\emptyset$ (with probability $\epsilon_x$), $\Delta'=\Delta.$
From  Part (i), $\EE[\Delta']^2 = \Psi'^2_x = \Psi^2_x \cdot \epsilon^2_x$.
Substituting these equations into Eqn.~\eqref{helps2}  establishes the result.

{\bf Remarks:}

In Proposition~\ref{pro1}, the two terms on the right of Eqn. (ii) for the variance ($\Phi'_x$) have a natural interpretation. The first term is the expected value of the conditional variance of $\Delta_\varphi'(S_x, x)$ given $I_x$. The second term is the variance of the conditional expectation of $\Delta_\varphi'(S_x,x)$ given $I_x$. Notice that the first term grows linearly with $\epsilon_x$, whereas the second term depends quadratically on $\epsilon_x$ (increasing from 0 at $\epsilon=0$ to a maximal value of 0.25 at $\epsilon_x=\frac{1}{2}$ then decreasing to zero again at $\epsilon_x=1$).

A third index also arises, which generalises an analogous index in \cite{ste07}.
Let 
$$\Delta''_\varphi(S_x,x) :=\varphi(S_x\cup I_x) - \varphi(S_x),$$
and let
$\Psi_x'' = \EE[\Delta''_\varphi(S_x, x)],$ and
$\Phi_x'' = \VV[\Delta''_\varphi (S_x, x)].$
Notice that $\Delta_\varphi'(S_x,x) + \Delta_\varphi''(S_x,x)= \Delta_\varphi(S_x,x)$, and so, by linearity of expectation, $\Psi'_x + \Psi''_x = \Psi_x$ and so $\Psi''_x = \Psi_x \cdot (1-\epsilon_x).$
Moreover, since
$$\Delta_\varphi''(S_x, x) 
= \begin{cases}
   \Delta_\varphi(S_x, x), & \mbox{ w.p. } 1-\epsilon_x,\\
   0, & \mbox{ w.p. } \epsilon_x,
\end{cases}$$
we obtain (by the same argument as for $\Phi'_x$) that
$\Phi''_x = \Phi_x\cdot (1-\epsilon_x) +\Psi_x^2\cdot (1-\epsilon)\cdot \epsilon_x$.

\section{PD Indices}
\label{phylo}

Following \cite{wik21},  consider a rooted phylogenetic tree $T=(V,E)$ with leaf set $X$, and with an assignment $\ell = [\ell_e]$ of a non-negative edge length to each edge of $T$. For a subset $S$ of $X$, let $\varphi_{(T, \ell)}(S)$ be the sum of the lengths of those edges of $T$ that lie on the path from the root to a leaf in $S$. As an example, for the phylogenetic tree $T$ shown in Fig.~\ref{fig:tree-network} and the set $S = \{1,3\}$, we obtain $\varphi_{(T, \ell)}(S) = \ell_1 + \ell_3 + \ell_6 + \ell_8 + \ell_9$.
\begin{figure}[ht]
    \centering
    \includegraphics[width=0.7\linewidth]{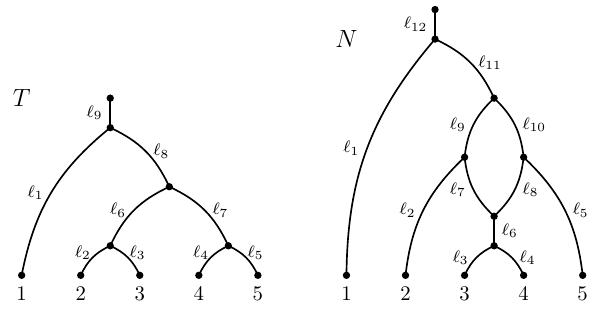}
   \caption{A rooted phylogenetic tree $T$ (left) and a rooted phylogenetic network $N$ (right) with leaf set $X = \{1,\ldots,5\}$ and non-negative edge lengths ($\ell_i \geq 0$ for each $i$). The edges are directed downwards.}
    \label{fig:tree-network}
\end{figure}

In this setting, the quantity $\varphi_{(T, \ell)}(S)$ is the widely-applied measure of PD of $S$, introduced by Daniel Faith in \cite{fai92}. Moreover, $\Psi$ and $\Psi'$ correspond to the HED and HEDGE indices, respectively, as studied in \cite{ste07} and applied in \cite{gum23}. Specifically, from \cite{ste07}, we have:

\begin{equation}
    \label{psieq}
    \Psi_x = \sum_{r=1}^{k_x} \ell(e_r)\cdot \left(\prod_{j \in C(e_r) \setminus\{x\}}\epsilon_j\right),
\end{equation}
where $e_1, \ldots, e_{k_x}$ are the edges on the path from the root of $T$ to the leaf $x$,  $C(e_j)$ is the set of leaves of $T$ that can be reached by a directed path from edge $e_j$ (note that we adopt the convention throughout this paper that $\prod_{j \in \emptyset}\epsilon_j=1$), and $\ell(e_r)$ is the length of edge $e_r$. Again referring to the tree $T$ in Fig.~\ref{fig:tree-network} and Eqn.~\eqref{psieq} we have, for example, 
\[\Psi_3 = \ell_9 \cdot \epsilon_1 \cdot \epsilon_2 \cdot \epsilon_4 \cdot \epsilon_5 + \ell_8 \cdot \epsilon_2 \cdot \epsilon_4 \cdot \epsilon_5 + \ell_6 \cdot \epsilon_2 + \ell_3.\]

Note that the values $\Psi_x$ are  quite different from PD indices such as `fair proportion' and `equal splits'. Those indices do not depend on $\epsilon_*$ values, and they satisfy the property that their values, summed over all species $x$, equal the total length of the tree (i.e. $\sum_e \ell_e$).

Observe also that (by Eqn.~\eqref{psieq} with the associated notation and  Proposition~\ref{pro1}(i)) we have: \begin{equation}
    \label{psieq2}
    \Psi'_x = \sum_{r=1}^{k_x} \ell(e_r)\cdot \left(\prod_{j \in C(e_r) }\epsilon_j\right).
\end{equation}

If instead of a phylogenetic tree, we consider a phylogenetic network (as shown on the right of Fig.~\ref{fig:tree-network}) then a natural choice for  $\varphi_{(N, \ell)}(S)$, and which generalises PD on trees, is  subnet-PD (the sum of the lengths of all edges of $N$ that lie on some path from the root to a leaf in $S$). We will describe how to calculate the quantities $\Psi_x$ and $\Phi_x$ for this index later, as this will follow directly from more general results that concern feature diversity.

\section{Extending PD to Feature Diversity}

The diversity function $\varphi$ has so far been arbitrary, except for one particular example, PD. We now consider an extension of PD to a more general setting (but still a special case of the framework introduced so far).

Following \cite{wik21}, let $\FF$ be a finite set whose elements we refer to as `features' (these  may correspond to traits, genomic  or morphological features, or any other discrete characteristics associated with the species). Suppose that each species $x \in X$ has an associated set of features $\alpha(x) \subseteq \FF$. Note that a given feature may be present in many species (or in just one or  none).  
Let $\nu:\FF \rightarrow {\RR}^{\geq 0}$ assign a non-negative value (i.e.  a weight or measure of significance) to each feature, and define $$\varphi_{(\FF, \nu)}: 2^X \rightarrow \RR^{\geq 0}$$ by:
$$\varphi_{(\FF, \nu)}(S) = \sum_{f \in \cup_{s \in S}\alpha(s)} \nu(f). $$ In words, $\varphi_{(\FF, \nu)}(S)$ is the sum of the values of all features present in at least one species in $S$. 

Note that PD on a rooted phylogenetic tree $T$ with leaf set $X$ is a special case of feature diversity if we take $\FF$ to be the set of edges of $T$ and $\nu$ to be the function that describes the length of each edge. In this setting, the feature set $\alpha(x)$ of a species $x \in X$ consists of the features corresponding to the edges from the root of $T$ to $x$.

\subsection{Properties of $\varphi_{(\FF, \nu)}$}
The function $\varphi_{(\FF, \nu)}$ satisfies the following properties: it is  monotone (i.e. $A\subseteq B$ implies $\varphi_{(\FF, \nu)}(A) \leq \varphi_{(\FF, \nu)}(B)$) and  submodular (i.e. for any two sets $A, B$, $\varphi_{(\FF, \nu)}(A\cup B)+\varphi_{(\FF, \nu)}(A\cap B)\leq \varphi_{(\FF, \nu)}(A)+\varphi_{(\FF, \nu)}(B)$). Also,   $\varphi_{(\FF, \nu)}(\emptyset) = 0$.
It follows that \[\varphi_{(\FF, \nu)}(S) \leq\sum_{x \in S}\sum_{f \in \alpha(x)} \nu(f),\] with equality precisely if the collection of feature sets $(\alpha(x): x\in S)$ are pairwise disjoint. 

Let us call any function  $\varphi:2^X \rightarrow \RR^{\geq 0}$ satisfying $\varphi(\emptyset)=0$ a {\em general diversity function}. If such a function can be written in the form $\varphi_{(\FF, \nu)}$ for some choice of $\FF$ and non-negative function $\nu$ it is referred to as a {\em weighted coverage function}. Such functions form a strict subset of the class of all monotone submodular functions $\varphi:2^X \rightarrow \RR^{\geq 0}$ with $\varphi(\emptyset)=0$.  
Not every  monotone submodular function can be described as a weighted coverage function. 
To see this, consider a set $X$ of size 3, and define $\varphi$ on $2^X$ by $\varphi(S) = |S|$ if $|S|\leq 2$ and $\varphi(X) = 2$. Then $\varphi$ is submodular, but a case analysis shows that it cannot be described as a weighted coverage function. 

A set $X$ of size 3 also suffices to show that (i) not every feature diversity function corresponds to PD on a tree and (ii) not every general diversity function is monotone and submodular.
To justify these two claims, let $X=\{x_1, x_2, x_3\}$.  For Part (i) take $\FF = \{f_1, f_2, f_3\}$, $\nu(f_i)=1$, for $i=1,2,3$, and feature assignment $\alpha(x_1) = \{f_1, f_2\}$,  $\alpha(x_2) = \{f_1, f_3\}$, and $\alpha(x_3) = \{f_2, f_3\}$.  Then $\varphi_{(\FF, \nu)}(S)=3$ for each set $S$ of size 2 as well as for the entire set $X$, and this cannot be realised by PD on any 3-leaf tree.

For Part (ii), take $\varphi(S)=|S|$ if $S \neq X$, and $\varphi(X)=4$. Take $A=\{x_1, x_2\}$ and $B=\{x_3\}$. Then submodularity requires $\varphi(A\cup B) + \varphi(A \cap B) \leq \varphi(A)+\varphi(B)$, however for our example this leads to the false inequality $4+0 \leq 2+1$.

The nesting of these classes is illustrated in Fig.~\ref{nest} (the diversity measures on phylogenetic networks are introduced in the second-to-last section.
Note that in Fig.~\ref{nest} `general diversity' includes but does not coincide with the {\em diversity} functions introduced by Bryant and Tupper~(\cite{bry12}) (the latter are monotone but may fail to be submodular).

\begin{figure}[htb]
    \centering
    \includegraphics[width=0.5\linewidth]{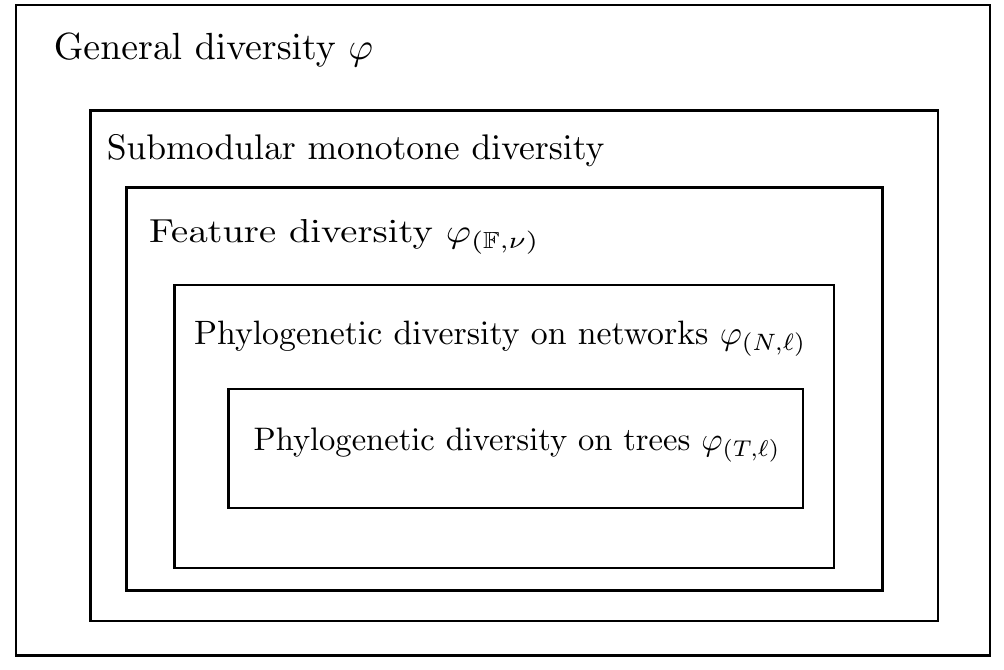}
    \caption{The nesting relationship between the different diversity measures described in this paper. We show later that two of these diversity classes (feature diversity and PD on networks) are equivalent as sets,  despite their quite different definitions. However, all other nestings are strict.}
    \label{nest}
\end{figure}

\bigskip

\subsection{Stochastic Properties of Feature Diversity Indices}

First, observe that if the random subset $S \subseteq X$ of species surviving an extinction event is modeled by the g-FOB model, then:
$$\EE[\varphi_{(\FF, \nu)}(S)] = \sum_{f \in \FF}\nu(f) \cdot (1-\prod_{x \in X: f \in \alpha(x)}\epsilon_x).$$
It can be shown that a  sufficient condition for $\varphi_{(\FF, \nu)}(S)/\sum_{f \in \FF} \nu(f)$ to be (asymptotically) close to its expected value   (i.e. $\EE[\varphi_{(\FF, \nu)}(S)]/\sum_{f \in \FF} \nu(f)$) as $|X|$ becomes large is provided by Proposition 2.3(c) of \cite{ove23}. Roughly speaking, this sufficient condition requires that the feature diversity of each species, divided by the average diversity of a randomly chosen species, is either bounded or it grows no faster than $\sqrt{|X|}.$

We turn now to the computation of $\Psi_x$ and $\Phi_x$ in the feature diversity setting.
The following result generalises Theorem 3.1 from \cite{ste07}.
\begin{prop}
\label{pro2}
Given a finite set of features $\FF$, a function $\alpha: X \rightarrow2^{\FF}$ that describes the features associated with each species,  and a weight function on features $\nu: \FF \rightarrow \RR^{\geq 0}$, let $\varphi= \varphi_{(\FF, \nu)}$. For features $f,f' \in \FF$ (allowing $f=f'$), let $$W_{x, ff'} = \{w \in X\setminus \{x\}: f \in \alpha(w) \mbox{ or }f' \in \alpha(w)\}.$$
Under a g-FOB model of extinction, the following equations hold for each species $x \in X$:

    $$\Psi_x = \sum_{f \in \alpha(x)} \nu(f) \prod_{y\in X\setminus \{x\}: f\in \alpha(y)} \epsilon_y$$
 and 
 $$\Phi_x  =  \left(\sum_{(f,f') \in \alpha(x)\times \alpha(x)} \nu(f)\nu(f') \prod_{y\in W_{x, ff'}} \epsilon_y\right) - \Psi_x^2.$$
    The corresponding values of $\Psi'_x$ and $\Phi'_x$ can then be obtained from Proposition~\ref{pro1}.

\end{prop}

{\em Proof.}
For any set $S_x \subseteq X\setminus\{x\}$, we have: $$\Delta_\varphi(S_x,x) = \sum_{f \in \alpha(x)\setminus \cup_{s \in S_x} \alpha(s)}\nu(f) = \sum_{f \in \alpha(x)} \nu(f) \cdot I_x(f),$$
    where $I_x(f)=1$ if $f \not\in \bigcup_{s \in S_x} \alpha(s)$, and $I_x(f)=0$ otherwise. 
 Under the g-FOB model we have:
$$\Psi_x = \EE[\Delta_\varphi(S_x,x)] = \sum_{f \in \alpha(x)} \nu(f) \cdot\PP\left(f \not\in \bigcup_{s \in S_x} \alpha(s)\right).$$
Here, the event $f \not\in \bigcup_{s \in S_x} \alpha(s)$ is precisely the event that all the species in $X\setminus\{x\}$ that contain feature $f$ go extinct. Under the g-FOB model, this event has probability $\prod_{y \in X\setminus \{x\}: f \in \alpha(y)} \epsilon_y$, which gives the claimed expression for $\Psi_x$.

For $\Phi_x$, we have:
\begin{equation}
    \label{vare}
    \Phi_x = \EE[\Delta_\varphi(S_x,x)^2] -\Psi_x^2
\end{equation}
and $$\Delta_\varphi(S_x,x)^2 = \sum_{(f,f') \in \alpha(x) \times \alpha(x)}  \nu(f)\nu(f') \cdot I_x(f)I_x(f').$$

Thus, $$\EE[\Delta_\varphi(S_x,x)^2] = \sum_{(f,f') \in \alpha(x) \times \alpha(x)}\nu(f)\nu(f') \cdot \PP(I_x(f)=1 \cap I_x(f')=1)$$
and
$$\PP(I_x(f)=1 \cap I_x(f')=1)=\prod_{y\in W_{x, ff'}}\epsilon_y.$$
The result now follows from Eqn.~\eqref{vare}.

\subsection{Example}
To illustrate the computation of $\Psi_x$, $\Psi'_x$,  $\Phi_x$, and $\Phi'_x$ in a simple setting involving feature diversity (but which cannot be modeled using PD on a tree), consider three species ($1,2,3$) with their associated feature sets $\alpha_1 = \{f_1, f_2\}$, $\alpha_2=\{f_2, f_3\}$, and $\alpha_3=\{f_1, f_3\}$. 
In that case, we have (for example),
$\Psi_1 = \nu(f_1)\epsilon_3+\nu(f_2)\epsilon_2$.

Next, for simplicity, set $\nu(f_i) =1$ for each of the three features and take the three species extinction probabilities  ($\epsilon_i$) to be a constant value $\epsilon$.
Applying Propositions~\ref{pro1} and \ref{pro2}, we obtain (for each species) the mean and variance quantities as follows:

\begin{equation}
    \label{nunu}
    \Psi_i = 2\epsilon \mbox{ and } \Psi'_i = 2\epsilon^2, \mbox{ and} 
\end{equation}
\begin{equation}
    \label{nunu2}
   \Phi_i = 2 \epsilon (1-\epsilon) \mbox{ and } \Phi_i' = 2 \epsilon^2 (1+\epsilon-2\epsilon^2) 
\end{equation}
In this simple example, the extinction probability that maximises the two variance quantities is 0.5 and $\sim$0.72, respectively.

Suppose  we now add a fourth species,  with $\alpha_4 = \{f_3, f_4\}$, and again set $\epsilon_4=\epsilon$, and $\nu(f_4)=1$. Unlike the other three species,  this fourth species has a feature that is `rare' (not present in the other species) and $\Psi_4 = 1+\epsilon^2$ and $\Psi'_4 = \epsilon(1+\epsilon^2)$. It follows that this fourth species has the highest $\Psi$  score for any $\epsilon \neq 1$ 
and the highest $\Psi'$ score for any $\epsilon \neq 0, 1$.

\subsection{The Variance $\Phi_x$ for a Tree with Edge Lengths}
As noted earlier, PD can be regarded as a special case of feature diversity by taking $\FF$ to be the set of edges of a tree $T$, and letting $\nu(e) = \ell(e)$ (the length of edge $e \in \FF$).  Then $\Psi_x$ is given by 
Eqn.~\eqref{psieq}, and the expression for $\Phi_x$ from Proposition~\ref{pro2} can be written more explicitly as follows. Let $e_1, \ldots, e_k$ be the edges on the path from the  root of $T$ to leaf $x$.
$$\Phi_x = \sum_{r=1}^{k_x}\sum_{s=1}^{k_x} \ell(e_r)\ell(e_s)\left(\prod_{j \in C(e_{\min\{r,s\}})\setminus\{x\}}\epsilon_j \right)- \Psi_x^2$$
where $C(e_*)$ denotes the set of leaves that can be reached by a directed path from $e_*$ (and again adopting the convention that $\prod_{j \in \emptyset}\epsilon_j=1$).
For example, returning to the tree $T$ depicted in Fig.~\ref{fig:tree-network}, we have already calculated
\[\Psi_3 = \ell_9 \cdot \epsilon_1 \cdot \epsilon_2 \cdot \epsilon_4 \cdot \epsilon_5 + \ell_8 \cdot \epsilon_2 \cdot \epsilon_4 \cdot \epsilon_5 + \ell_6 \cdot \epsilon_2 + \ell_3.\]
This leads to
\begin{align*}
    \Phi_3 &= \bigg( \ell_9^2 \cdot \epsilon_1 \cdot \epsilon_2 \cdot \epsilon_4 \cdot \epsilon_5 + \ell_8^2 \cdot \epsilon_2 \cdot \epsilon_4 \cdot \epsilon_5 + \ell_6^2 \cdot \epsilon_2 + \ell_3^2 \cdot 1 \phantom{\bigg)} \\
    &\quad \phantom{\bigg(} + 2 \cdot \ell_9 \cdot \ell_8 \cdot \epsilon_1 \cdot \epsilon_2 \cdot \epsilon_4 \cdot \epsilon_5 + 2 \cdot \ell_9 \cdot \ell_6 \cdot \epsilon_1 \cdot \epsilon_2 \cdot \epsilon_4 \cdot \epsilon_5 +  2 \cdot \ell_9 \cdot \ell_3 \cdot \epsilon_1 \cdot \epsilon_2 \cdot \epsilon_4 \cdot \epsilon_5  \phantom{\bigg)} \\
    &\quad \phantom{\bigg(} +  2 \cdot \ell_8 \cdot \ell_6 \cdot \epsilon_2 \cdot \epsilon_4 \cdot \epsilon_5 + 2 \cdot \ell_8 \cdot \ell_3 \cdot \epsilon_2 \cdot \epsilon_4 \cdot \epsilon_5  + 2 \cdot \ell_6 \cdot \ell_3 \cdot \epsilon_2 \bigg)\\
    &\quad- (\ell_9 \cdot \epsilon_1 \cdot \epsilon_2 \cdot \epsilon_4 \cdot \epsilon_5 + \ell_8 \cdot \epsilon_2 \cdot \epsilon_4 \cdot \epsilon_5 + \ell_6 \cdot \epsilon_2 + \ell_3)^2, \\
\end{align*}
which, for example, is approximately equal to $0.714844$ if $\ell_i = 1$ and $\epsilon_i = 0.5$ for all $i$.

\subsection{$\Psi_x$ as a Shapley Value for Feature Diversity}

In the simple field-of-bullets model, all extinction probabilities are equal (i.e.  $\epsilon_x = \epsilon$ for all species $x \in X$).
If this extinction risk $\epsilon$ is assumed to be uniformly distributed between $0$ and $1$, then the expected value of $\Psi_x$ with respect to this prior, denoted $\overline{\Psi_x}$, is given by:
\begin{equation}
\label{psi}
    \overline{\Psi_x} = \sum_{f \in \alpha(x)} \nu(f)/N(f),
\end{equation}
where $N(f)$ is the total number of species in $X$ that contain the feature $f$. The proof of this claim follows the same argument as that used to establish Eqn. (3.7) in \cite{ste25}.

The right-hand side of Eqn.~\eqref{psi} is precisely the Shapley value of species $x$ in the cooperative game in which the total feature diversity score is allocated among species according to their marginal contributions, as stated in Proposition 2 of \cite{wik21} (a more direct proof uses the generalisation of Theorem 1 of \cite{cor18} described in the conclusion section of that paper).
Thus Eqn.~\eqref{psi} provides a direct interpretation of $\overline{\Psi}_x$ as a Shapley value, where the features are shared out `fairly'.

Eqn.~\eqref{psi} also leads to a further identity. 
Let $Av(\overline{\Psi_x})$ be the average value of $\overline{\Psi_x}$ over all $n$ species. Then $Av(\overline{\Psi_x})$ is independent of how the features are distributed across species; it depends solely on the number $n$ of species and the total feature diversity present. We state this more formally as follows.

\begin{prop}
For the simple field-of-bullets model described above,
    $$Av(\overline{\Psi_x}) =\frac{1}{n} \sum_{f: N(f)>0}\nu(f).$$
\end{prop}
{\em Proof.}
$$Av(\overline{\Psi_x}) =  \frac{1}{n} \sum_x \sum_{f\in \alpha(x)}\frac{\nu(f)}{N(f)} =  \frac{1}{n}  \sum_{f: N(f)>0}\nu(f)\cdot \sum_{x: f \in \alpha(x)} \frac{1}{N(f)},$$
and the last summation term equals 1 for each feature $f$.

\subsection{Example Application}

We apply the calculation of $\Psi_x$ and $\Psi'_x$ and their variances ($\Phi_x$ and $\Phi'_x$) on the phylogeny of the Order Crocodilia (alligators, crocodiles, and gharials), a small, isolated clade of endangered reptiles. Over half of the 27 species are of conservation concern, with 7 listed as critically endangered by the International Union for Conservation of Nature (IUCN). We took a time-calibrated phylogeny of the group and each species' corresponding conservation status (from \cite{col20}) and converted those statuses to $\epsilon_x$, following \cite{gum23}. We then calculated the values of $\Psi_x$ and $\Psi'_x$ and their associated standard deviations $\sqrt{\Phi_x}$ and $\sqrt{\Phi'_x}$. These values are presented in Fig.~\ref{fig:croc}.

\begin{figure}[ht]
    \centering
    \includegraphics[width=1.0\linewidth]{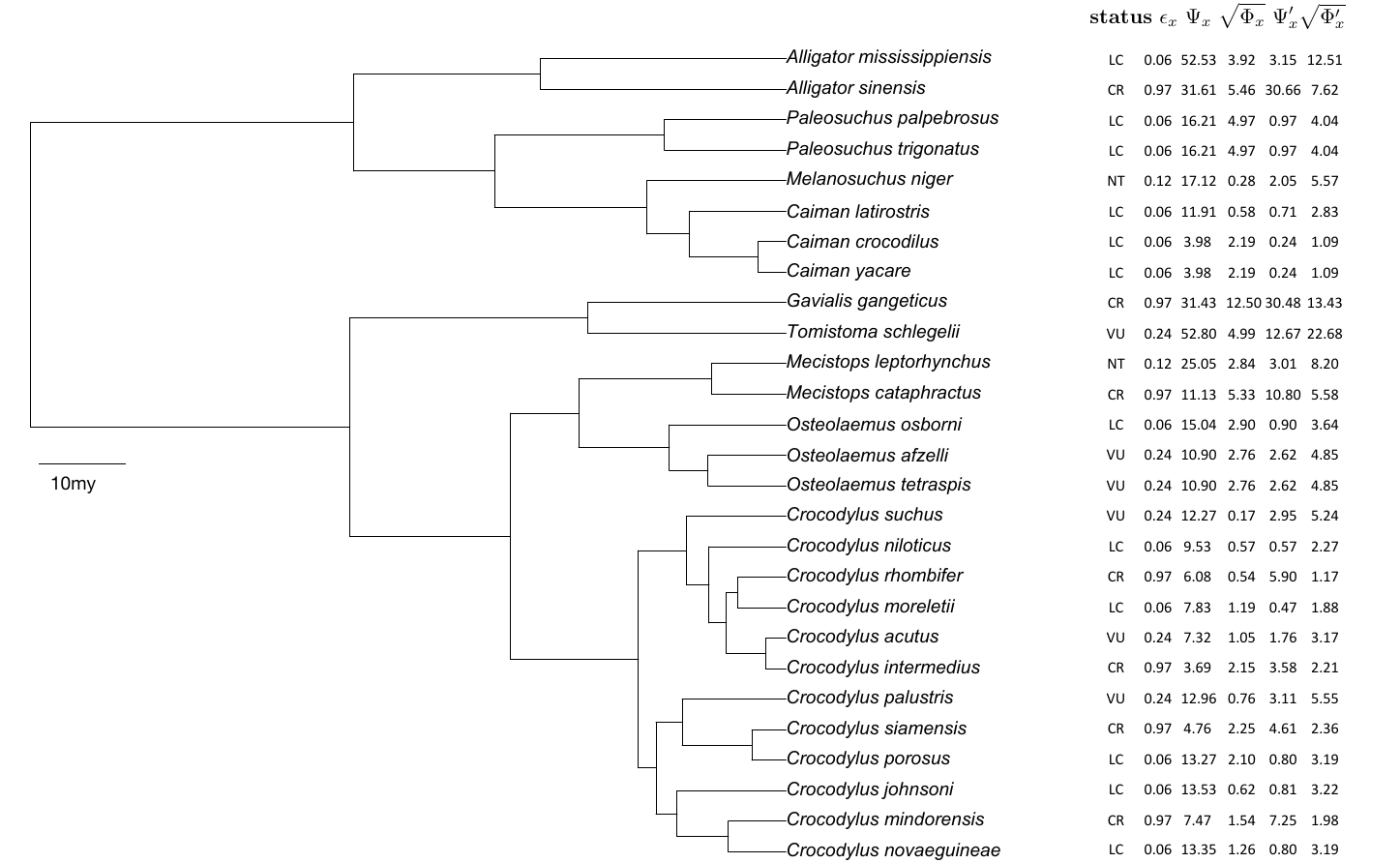}
     \caption{The expected values ($\Psi_x$ and $\Psi'_x)$ and standard deviations ($\sqrt{\Phi_x}$ and $\sqrt{\Phi'_x}$) of 
     $\Delta_\varphi(S_x,x)$
     and $\Delta'_\varphi(S_x,x)$, respectively, for the Crocodilian data-set. Both statistics vary widely across the species. See text for further details.}
     
    \label{fig:croc}
\end{figure}

There are several things to note. The values $\Psi$ and $\Psi'$ and their standard deviations all show a broad range across species, driven by their position on the tree and both their own and their close relatives' $\epsilon$ values. For instance, the first species, {\em Alligator mississippiensis}, is a widely-distributed Least Concern species that nevertheless carries a lot of genetic responsibility for its highly endangered sister species; this is captured in its particularly high  $\Psi$ score. Contrast this with its sister species {\em A. sinensis}, whose $\Psi$ score is much smaller; the $\Psi'$ scores are in reverse order due to the much larger extinction probability for {\em A. sinensis}.

Note also the relatively large value of the standard deviation value $\Phi'_x$ for {\em A. mississippiensis}. In general, species such as this with a low $\epsilon$ value  (and thus a low $\Psi'$ value) show standard deviations for $\Psi'$ that are higher than their expected value:
Proposition~\ref{pro1}(ii) implies that  $\sqrt{\Phi'_x} \geq \Psi'_x \sqrt\frac{1-\epsilon_x}{\epsilon_x}$. Thus,  if $\epsilon_x$ is less than $0.5$, then  the standard deviation of $\Delta'_\varphi(S_x,x)$ is larger than its mean (the converse need not hold).  Note also that $\Psi'_x$ and $\sqrt{\Phi'_x}$ both converge to 0 as $\epsilon_x \rightarrow 0$.

 For a species with a low $\epsilon$ value, by far the most common outcome of protecting it is no gain at all, but there remains the low probability that the intervention will produce high gains. For species at high risk of extinction, the outcome of intervention is more assured.  Because low-$\Psi'$ species (like the {\em Caiman spp.}) would not generally be considered primary targets for conservation~(\cite{gum23}),  the fact that the variance in their $\Phi'_x$ is large may not be particularly worrisome, but this does bear keeping in mind: lowering their $\epsilon_x$ values via conservation could avoid unexpected large losses of biodiversity.  

Fig.~\ref{fig:moments} provides some further insight into the behaviour of $\Psi'_x$ and its associated standard deviation  $\sqrt{\Phi'_x}$, again by considering the two alligator species that appear at the top of the tree in Fig.~\ref{fig:croc} ({\em A. mississippiensis} and {\em A. sinsensis}).  In the first case (for $x_1$= {\em mississipiensis}), the mean is small, and the standard deviation is large; in the second (for $x_2$ = {\em sinensis})), the reverse occurs (note that in Fig.~\ref{fig:moments}, there are a few other extremely small values (not shown) that have no significant effect). These two observations are explained by the way the most common outcomes (with probabilities of 0.94 and 0.91) move from 0 (for {\em mississipiensis}) to closer to the middle of the distribution (for {\em sinensis}) in Fig.~\ref{fig:moments}.

\begin{figure}[ht]
    \centering
    \includegraphics[width=0.6\linewidth]{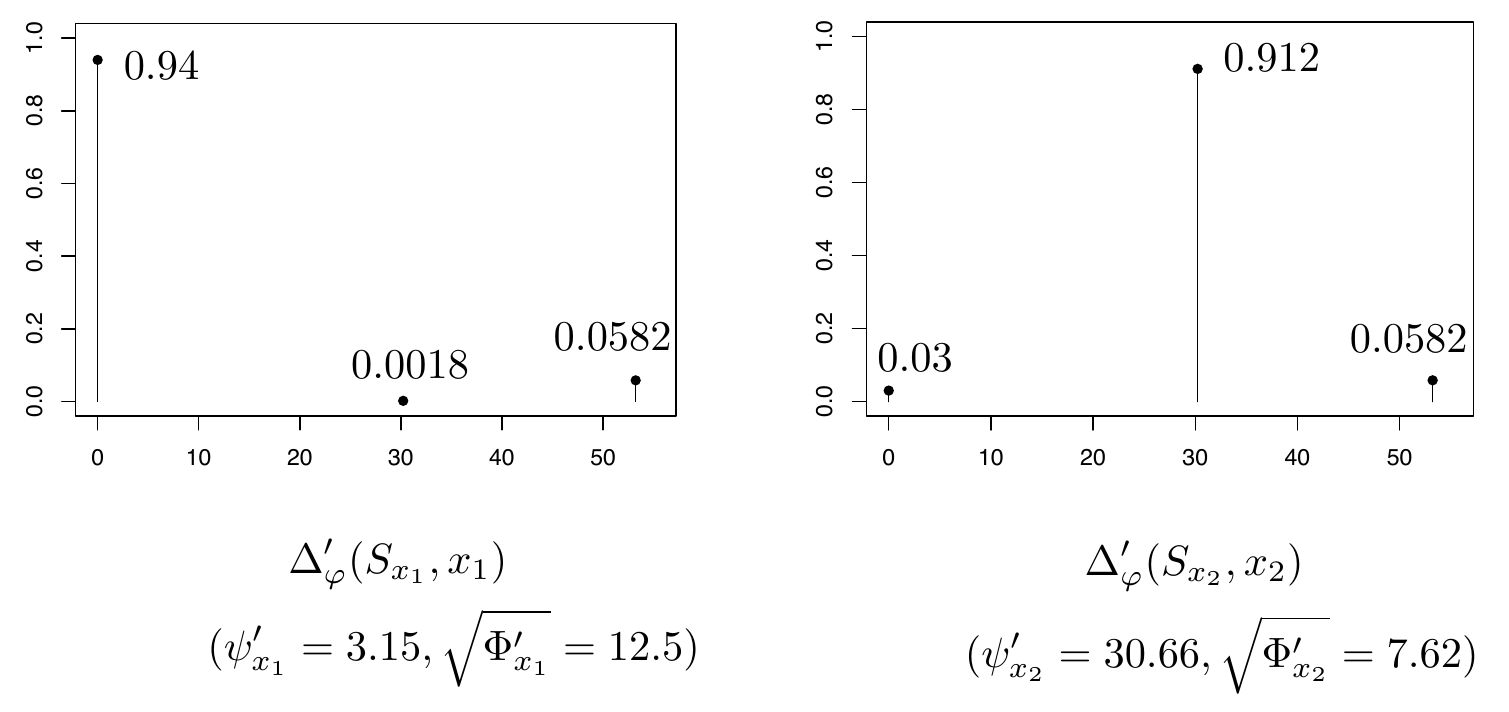}
     \caption{The distribution of $\Delta'_\varphi(S_x,x)$ for the two alligator 
     species at the top of Fig.~\ref{fig:croc} ($x_1=$ {\em Alligator mississippiensis} and $x_2=$ {\em Alligator sinensis}), together with their associated means and standard deviations. The $x$-axis is in units of millions of years.}
    \label{fig:moments}
\end{figure}

\subsection{Feature diversity}
Using this same collection of 27 Crocodilian species, we also computed diversity indices based on a small set of eight features (size, breeding behaviour, and ecological traits) taken from the literature, rather than based on a phylogenetic tree with branch lengths. We followed Proposition~\ref{pro2}, setting the value $\nu(f)$ of each feature $f$ equal to the same value (without loss of generality, $\nu(f)=1$), and so the diversity indices are simple polynomials of the $\epsilon$-values of non-focal species sharing features with the focal species. These products quickly become small. We therefore report the $z$-scores for the diversity indices (such that the sum of each diversity index across the species is 0, and any fixed $\nu(f)$ will produce the same results) rather than their raw values, cognizant that the raw scores may not be drawn from a Gaussian distribution.

This  dataset of haphazard crocodilian features is too small to glean much from the $z$-scores of the diversity indices and their variances. As expected, the only three species ({\em Crocodilus  niloticus}, {\em C. porosus}, and {\em C. acutus}) that express the rarest feature of living in salt water score highest for $\Psi$ (and for $\Psi'$); variation among these three is dictated by the product of the $\epsilon$-values of the other two species expressing the feature, such that {\em C. porosus} is the most replaceable among them. Generally, species that express more features are more irreplaceable (for example, the top 14 most irreplaceable species have an average of 3.1 features, while the bottom 13 express 1.9 features). It is also worth noting that none of the three {\em Crocodilus} rare-feature species are particularly irreplaceable if we use the phylogenetic tree to measure $\Psi$ (Fig.~\ref{fig:croc}).  We await better datasets to investigate the properties of the tree-free feature diversity indices in more detail.

\begin{table}[htbp]
\centering
%\caption{Crocodilian species, trait states, and network metrics.}
\label{tab:crocodiles_8_features}
\resizebox{\textwidth}{!}{%
\begin{tabular}{lccccccccccccc}
\hline
Scientific Name & $\epsilon_x$ & LB & HN & LCA & CN & NT & M & DS  & HS & $z-\Psi_x$ & $z-\Psi'_x$ & $z-\Phi_x$ & $z-\Phi'_x$ \\
\hline
{\em Alligator mississippiensis} & 0.06 & 1 & 0 & 1 & 0 & 0 & 0 & 0 & 0 & -0.311884 & -0.350856 & -0.312235 & -0.350891 \\
{\em Alligator sinensis} & 0.97 & 0 & 0 & 0 & 0 & 1 & 0 & 1 & 1 & -0.312676 & -0.343715 & -0.313038 & -0.343749 \\
{\em Caiman crocodilus} & 0.06 & 0 & 0 & 1 & 0 & 0 & 0 & 0 & 0 & -0.313289 & -0.352035 & -0.313660 & -0.352071 \\
{\em Caiman latirostris} & 0.06 & 0 & 0 & 1 & 1 & 0 & 0 & 0 & 0 & -0.300838 & -0.341588 & -0.301042 & -0.341626 \\
{\em Caiman yacare} & 0.06 & 0 & 0 & 0 & 0 & 0 & 0 & 1 & 1 & -0.309598 & -0.348938 & -0.309919 & -0.348974 \\
{\em Crocodylus acutus} & 0.24 & 1 & 1 & 1 & 1 & 0 & 1 & 0 & 0 & 0.618936 & 2.776612 & 0.628735 & 2.776941 \\
{\em Crocodylus intermedius} & 0.97 & 1 & 1 & 1 & 0 & 0 & 0 & 0 & 0 & -0.313161 & -0.350290 & -0.313529 & -0.350325 \\
{\em Crocodylus johnsoni} & 0.06 & 0 & 1 & 0 & 1 & 1 & 0 & 0 & 1 & -0.293205 & -0.335184 & -0.293025 & -0.334988 \\
{\em Crocodylus mindorensis} & 0.97 & 0 & 0 & 0 & 0 & 0 & 0 & 0 & 1 & -0.313244 & -0.351421 & -0.313614 & -0.351457 \\
{\em Crocodylus moreletii} & 0.06 & 0 & 0 & 0 & 0 & 0 & 0 & 0 & 1 & -0.312558 & -0.351421 & -0.312918 & -0.351457 \\
{\em Crocodylus niloticus} & 0.06 & 1 & 1 & 0 & 1 & 1 & 1 & 1 & 0 & 3.424797 & 2.784318 & 3.424462 & 2.785023 \\
{\em Crocodylus novaeguineae} & 0.06 & 0 & 0 & 0 & 0 & 0 & 0 & 0 & 1 & -0.312558 & -0.351421 & -0.312918 & -0.351457 \\
{\em Crocodylus palustris} & 0.24 & 1 & 1 & 1 & 1 & 0 & 0 & 0 & 0 & -0.309657 & -0.339842 & -0.309957 & -0.339810 \\
{\em Crocodylus porosus} & 0.06 & 1 & 0 & 0 & 0 & 0 & 1 & 0 & 0 & 3.402485 & 2.765598 & 3.400352 & 2.764557 \\
{\em Crocodylus rhombifer} & 0.97 & 1 & 1 & 1 & 0 & 0 & 0 & 0 & 0 & -0.313161 & -0.350290 & -0.313529 & -0.350325 \\
{\em Crocodylus siamensis} & 0.97 & 0 & 0 & 0 & 0 & 0 & 0 & 1 & 1 & -0.313061 & -0.348938 & -0.313429 & -0.348974 \\
{\em Crocodylus suchus} & 0.24 & 0 & 1 & 1 & 1 & 1 & 0 & 0 & 1 & -0.308269 & -0.335184 & -0.308501 & -0.334988 \\
{\em Gavialis gangeticus} & 0.97 & 1 & 1 & 0 & 1 & 0 & 0 & 1 & 0 & -0.312208 & -0.337359 & -0.312558 & -0.337327 \\
{\em Mecistops cataphractus} & 0.97 & 0 & 0 & 0 & 0 & 0 & 0 & 1 & 0 & -0.313107 & -0.349553 & -0.313474 & -0.349588 \\
{\em Mecistops leptorhynchus} & 0.12 & 0 & 0 & 0 & 0 & 1 & 0 & 0 & 0 & -0.310177 & -0.346812 & -0.310505 & -0.346848 \\
{\em Melanosuchus niger} & 0.12 & 1 & 0 & 1 & 0 & 0 & 0 & 0 & 0 & -0.312587 & -0.350856 & -0.312947 & -0.350891 \\
{\em Osteolaemus afzelli} & 0.24 & 0 & 1 & 1 & 0 & 1 & 0 & 0 & 0 & -0.311565 & -0.346246 & -0.311872 & -0.346150 \\
{\em Osteolaemus osborni} & 0.06 & 0 & 1 & 0 & 0 & 1 & 0 & 1 & 0 & -0.303430 & -0.343763 & -0.303506 & -0.343666 \\
{\em Osteolaemus tetraspis} & 0.24 & 0 & 1 & 1 & 0 & 0 & 0 & 0 & 0 & -0.313121 & -0.351469 & -0.313489 & -0.351505 \\
{\em Paleosuchus palpebrosus} & 0.06 & 0 & 0 & 1 & 0 & 0 & 0 & 1 & 1 & -0.309598 & -0.348938 & -0.309918 & -0.348974 \\
{\em Paleosuchus trigonatus} & 0.06 & 0 & 0 & 1 & 0 & 0 & 0 & 1 & 0 & -0.310330 & -0.349552 & -0.310661 & -0.349588 \\
{\em Tomistoma schlegelii} & 0.24 & 1 & 0 & 0 & 0 & 0 & 0 & 0 & 0 & -0.312938 & -0.350856 & -0.313304 & -0.350892 \\
\hline
\end{tabular}}
\caption{Crocodilian species traits and  $z$-scores of feature diversity statistics. Abbreviations: LB=Large-bodied,  HN=Hole Nesting, LCA=Loss of Cathemeral Activity, CN=Communal Nesting, NT=Not Terrestrial, M=Marine habitat use, DS=Diet Specialism, HS=Habitat Specialism.}
\label{tab:croc_traits}
\end{table}

\section{Diversity Measures on Phylogenetic Networks}
\label{nets}

Although PD was originally defined for phylogenetic trees, different extensions to phylogenetic networks have been proposed (see, e.g., \cite{bor22, kon25,  cor18, wik18}). We review some of them below.

A \emph{rooted phylogenetic network} $N$ on $X$ is a rooted directed acyclic graph with no parallel edges satisfying the following properties:
\begin{enumerate}[(i)]
    \item The (unique) root has in-degree zero and out-degree one;
    \item A vertex with out-degree zero has in-degree one, and the set of vertices with out-degree zero is $X$ (these are the \emph{leaves}); and
    \item All other vertices have either in-degree one and out-degree at least two (these are called \emph{tree vertices)}, or in-degree at least two and out-degree one (these are called \emph{reticulations}).
\end{enumerate}

An example of a rooted phylogenetic network is shown in Fig.~\ref{fig:tree-network}.

Furthermore, we assume that each edge of $N$ has a non-negative real-valued length. Thus, if $E$ denotes the edge set of $N$, there is a function $\ell: E \rightarrow \mathbb{R}^{\geq 0}$ for which each edge $e$ of $N$ is assigned the weight $\ell(e)$.

As with trees, we formally assign a unique feature $f_e$ to each edge $e$,  and let $\nu(f_e)$ equal the length of the edge $e$. Then $\varphi_{(\FF, \nu)}$ is  the `subnet phylogenetic diversity' (subnet-PD) measure described in \cite{wik18} and studied further in \cite{cor18} and, under the name of `AllPaths-PD', in \cite{bor22}. Formally, given a phylogenetic network $N$ on $X$ and a subset $S \subseteq X$ of taxa, the subnet-PD of $S$ is the sum of lengths of all edges of $N$ that lie on some path from the root to a leaf in $S$. For example, considering the network $N$ in Fig.~\ref{fig:tree-network} and $S = \{3,4\}$, the subnet phylogenetic diversity of $S$ equals $\ell_3 + \ell_4 + \ell_6 + \ell_7 + \ell_8 + \ell_9 + \ell_{10} + \ell_{11} + \ell_{12}$.

\paragraph{Other notions of phylogenetic diversity on networks.}
While subnet-PD provides a natural extension of Faith's PD to phylogenetic networks, it is based on the assumption that all features arising along any root-to-leaf path are inherited by the corresponding taxon. In particular, in the presence of reticulation (hybridisation) events, this implies that a hybrid species inherits the full set of features from each of its parents.

Alternative models relax the assumption that all features are fully inherited at reticulations by allowing partial inheritance. In these models, each incoming edge of a reticulation is assigned an inheritance probability, and the contribution of a feature to a descendant taxon is weighted accordingly. 

More precisely, following \cite{bor22}, for each edge $e=(u,v)$ directed into a reticulation $v$, let $p(e) \in [0,1]$ denote the proportion of features at the parent vertex $u$ that are inherited by the child vertex $v$. For a subset $S \subseteq X$ and an edge $e = (u,v)$ of $N$, let $\gamma(S,e)$ denote the proportion of features present at $v$ that are inherited by at least one taxon in $S$. Equivalently, $\gamma(S,e)$ can be interpreted as the probability that a feature arising on edge $e$ is inherited by some taxon in $S$.

The resulting diversity measure, termed \emph{Network-PD} in \cite{bor22}, is defined by
\[ \textup{Network-PD}_{N, p}(S) = \sum\limits_{e \in N} \gamma(S,e) \cdot \ell(e).\]
The values $\gamma(S,e)$ can be computed recursively in a bottom-up manner (for details, see \cite{bor22}). We note that when $p(e)\equiv1$ for all reticulation edges, Network-PD reduces to AllPaths-PD (and hence subnet-PD).

In some settings, it is natural to assume that at a reticulation with incoming edges $e_1,e_2,\ldots, e_k$, the inheritance proportions satisfy $p(e_1) + p(e_2) + \ldots + p(e_k)=1$. Under this assumption, upper and lower bounds for Network-PD can be obtained via the notions of \emph{MaxWeightTree-PD} and \emph{MinWeightTree-PD}, introduced in \cite{bor22}. Informally, these measures quantify diversity by considering trees displayed by the network (obtained by selecting one incoming edge at each reticulation), and then taking the maximum or minimum phylogenetic diversity over all such trees. A related measure, \emph{AverageTree-PD} introduced by \cite{ier25}, instead computes the expected diversity over these trees, weighing each tree by its probability (obtained as the product of the inheritance probabilities associated with the edges leading into reticulations that are retained in the displayed tree). We refer the reader to \cite{bor22,ier25} for more formal definitions and further details.

We note that several of the above diversity measures have recently been implemented in \texttt{PaNDA} (Phylogenetic Network Diversity Algorithms), a Python software package with an interactive graphical user interface for exploring, visualising, and optimising diversity in phylogenetic networks (\cite{hol25}). This development is promising, as much of the existing work on network-based diversity measures has been driven primarily by mathematical and theoretical computer science considerations (see, e.g., \cite{bor22,cor24,ier25,ier25b,jon23} and references therein). Consequently, an important direction for future research is to assess the practical performance of these measures and to determine which notions of phylogenetic diversity on networks are most informative and biologically meaningful in empirical settings.

\medskip
In what follows, however, we restrict attention to the subnet-PD setting, as it admits a direct correspondence with feature diversity indices.
Consider an arbitrary assignment of features to the species in $X$. It is easy to see that such an assignment cannot always be realised by an assignment of features to the edges of a tree (where each feature arises just once on some edge and is inherited by all species below it). 
For example, if the leaf species $1,2,3$ have the feature sets $\{f_1, f_2\}, \{f_2, f_3\}$, and $\{f_1, f_3\}$, then no tree can represent this assignment of features under these assumptions.

This leads to the following question: can every assignment of features to the species in $X$ be realised by some phylogenetic network (i.e. the set of features present in the leaf species $x \in X$ is the union of the features present on edges that lead to  leaf $x$ and, again, each feature arises just once on some edge and is never lost). 

The answer to this questions is easily seen to be `yes'. To see why, let 
$X = \{1, \ldots, n\}$ and let  $X_f = \{x \in X: f \in F_x\}$ denote the subset of taxa that have the feature $f$. Also, let  $\mathcal{C}_\mathcal{F} = \{X_f: f \in \mathcal{F}\}$ be the collection of the sets $X_f$. Construct the Hasse diagram for $\mathcal{C}_\mathcal{F} \cup X \cup \{1\} \cup \ldots \cup  \{n\}$, inserting additional edges to ensure that each reticulation has out-degree 1 to obtain a phylogenetic network $N_{\FF}$. Then the underlying feature distribution can be realised on $N_{\FF}$ by letting feature $f$ arise on the edge directed into the vertex representing the set $X_f$ (i.e. whose cluster is $X_f$). If this is a reticulation, it can be placed on any of the incoming edges or, alternatively, on the outgoing edge.
Returning to the example above, suppose that $X = \{1,2,3\}$ with $\alpha_1 = \{f_1, f_2\}$, $\alpha_2 = \{f_2, f_3\}$, and $\alpha_3 = \{f_1, f_3\}$. Then, $X_{f_1} = \{1,3\}$, $X_{f_2} = \{1,2\}$, and $X_{f_3} = \{2,3\}$, and the rooted phylogenetic network $N_\FF$ shown in Fig.~\ref{fig:NF} realises this feature distribution.

\begin{figure}[ht]
    \centering
    \includegraphics[width=0.36\linewidth]{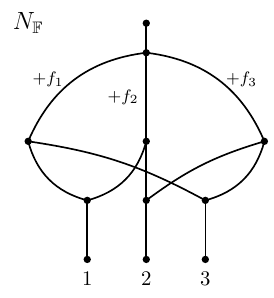}
     \caption{A rooted phylogenetic network $N_\FF$ realising the feature distribution $\alpha_1 = \{f_1, f_2\}$, $\alpha_2 = \{f_2, f_3\}$, and $\alpha_3 = \{f_1, f_3\}$.}
    \label{fig:NF}
\end{figure}

The following result uses this construction to show that the class of  subnet-PD indices on phylogenetic networks is precisely the same class of diversity indices as the class of feature diversity indices.

\begin{prop}
\label{identical}
An arbitrary diversity index $\varphi$ on $X$ can be represented as a feature diversity index $\varphi_{(\FF, \nu)}$ if and only if $\varphi$ is the subnet-PD  of some edge-weighted phylogenetic network on leaf set $X$.
\end{prop}

{\em Proof.}
The proof of the `if' direction was outlined earlier in this section. For the `only if' direction, suppose that $\varphi=\varphi_{(\FF, \nu)}$. Construct the network $N_{\FF}$ described above (with the assignment of features to the edges of $N_{\FF}$) and set the length of edge $e$ of $N_{\FF}$ equal to the sum of the $\nu(f)$ values that have feature $f$ assigned to that edge (where the length of $e$ is zero if no feature arises on $e$). Then the subnet diversity of this resulting edge-weighted network is $\varphi$.

Proposition~\ref{identical} implies that in Fig.~\ref{nest} the classes `Feature diversity' and `Phylogenetic diversity on networks' are actually equivalent as sets, despite their somewhat different descriptions.

\section{Concluding Comments}

The amount of additional phylogenetic diversity that can be saved if a particular species $x$ is protected against extinction generally depends on several factors, including (i) the position of species $x$ in the phylogenetic tree and the lengths of the branches leading to it, (ii) the fate of other species (captured in their assigned $\epsilon$-values), and (iii) the probability ($\epsilon_x$) that species $x$ will become extinct if it is not protected. Under the simple (generalised) field-of-bullets model, the expected value of this additional phylogenetic diversity (i.e. the expected gain following conservation) can be easily computed, and corresponds to the EDGE2 score (i.e. the HEDGE value) for this species~(\cite{gum23}). Moreover, it is the product of two terms, one of which is simply the extinction probability of species $x$ itself ($\epsilon_x$)~(\cite{ste07}). In this paper, we have generalised this setting to arbitrary diversity functions (which may apply to trees or phylogenetic networks, or be completely independent of phylogenetic considerations). We observe that the same factorisation applies for the expected value of the additional diversity, which allows us to provide an exact expression for the standard deviation of this random variable.

We also explored the relationships between feature diversity, phylogenetic diversity, and other types of diversity measures (on trees and networks). Feature diversity is closely related to phylogenetic diversity on trees provided that features arise only once in the tree and are never lost. However, without these last assumptions there is no direct connection between the two diversity measures~(\cite{wik21, ros23}).  This may help to explain some empirical findings, particularly those that find weak correlations between tree-based and feature-based measures of diversity~(\cite{kel14, maz18}; see also \cite{hah24}). On the other hand, when we moved from phylogenetic trees to phylogenetic networks we established a much more direct connection between subnet diversity  and feature diversity (Proposition~\ref{identical}). For feature diversity, we also derived more transparent  expressions for the expected values ($\Psi, \Psi'$) and their corresponding variances $(\Phi, \Phi')$.

We then applied our results to the phylogenetic tree of, and a small set of features for, extant Crocodilians, to derive the $\Psi$ and $\Psi'$ scores and compare these alongside the corresponding standard deviation quantities ($\sqrt{\Phi}$ and $\sqrt{\Phi'}$). We discussed their behaviour as the extinction probability ($\epsilon$) varies, and noted that as random variables, these metrics are perforce estimated with considerable imprecision. Armed with estimates of that imprecision (the expected variances), practitioners can now ask if two species really differ in, e.g. their EDGE2 scores such that they should be given differential attention. The observation that species with high probabilities of extinction (high $\epsilon$) tend to have more precise diversity metric estimates (other things being equal) offers support for choosing among them to focus conservation effort. In contrast, the knowledge that species with low probabilities of extinction (low $\epsilon$) tend to have imprecisely estimated scores is also helpful: for species that score low for $\Psi$ (and so $\Psi'$), relative ranks will be imprecise and so unhelpful. For species with low probabilities of extinction but that are nonetheless -- because close relatives have high $\epsilon$ --  irreplaceable, imprecision may warrant an opposite stance: in a maximum-loss avoidance strategy setting, forward-looking activity to drop $\epsilon$ to zero may be warranted (see also \cite{fai15}).

With reference to Crocodylia, our results can complement work done by the ``Edge of Existence" programme at the Zoological Society of London, which has previously highlighted the critically endangered gharial ({\em Gavialis gangeticus}) as a top-ranked EDGE vertebrate species.  The other high-ranking species in the Crocodilia using $\Psi'$ in Fig.~\ref{fig:croc} is {\em A. sinensis}, a critically endangered species under strong conservation management in China. Each is in a two-species family, and each of their distantly-related sister-taxa  have top-ranked $\Psi$ (irreplaceability) scores: based on our analysis, the false gharial {\em Tomistoma schlegelii} might warrant particular attention because it has the third highest EDGE2 ($\Psi'$) score in the order and its score is by far the most imprecise. Were the gharial to go extinct the expected gain accrued by conserving the false gharial could be very high indeed. Highlighting this aspect of its value could lead to more conservation attention.

In future work, it would be interesting to explore extensions of our results to processes more general than the g-FOB model, as well as applying general feature diversity measures to conservation-relevant feature datasets and settings where networks rather than trees better represent relationships.  

Regarding the first extension (to allow  extinction risk to depend on the extinction of other species),  the earlier expression for $\Psi_x$ given in Proposition~\ref{pro2} needs to be modified.   First observe that if  $E(y)$ is the event that species $y$ becomes extinct, then the earlier expression for  $\Psi_x = \EE[\Delta'_\varphi(S_x, x)]$ becomes 
$$\Psi_x =\sum_{f \in \alpha(x)}\nu(f) \cdot \PP\left(\bigcap_{y \in X\setminus \{x\}: f\in \alpha(y)}E(y)\right).$$
Now, if feature $f \in \alpha(x)$ is present in at most one other species than $x$, then 
the term $\PP\left(\bigcap_{y \in X\setminus \{x\}: f\in \alpha(y)}E(y)\right)$ is the same as it was in the earlier analysis where independence was assumed.
On the other hand, if $f$ is a feature present in at least two other species of $X\setminus \{x\}$, then for any ordering of these species that have feature $f$, say,  $y_1, y_2, \ldots, y_k$ (where $k$ depends on $f$), the term $\PP\left(\bigcap_{y \in X\setminus \{x\}: f\in \alpha(y)}E(y)\right)$ equals $\PP(E(y_1))$ multiplied by the product of the following conditional probabilities: $\prod_{j=2}^{k} \PP\left(E(y_j)|\bigcap_{i=1}^{j-1} E(y_i)\right)$. In particular, if species extinction events tend to increase the probability that other species will become extinct, then this product term will be larger than the product of the unconditional extinction probabilities  (the $\epsilon_y = \PP(E(y))$ values), and this dependency will increase the value of $\Psi_x$ over what it would have been for independent extinctions.

Finally, if our goal is to maximise the conservation of total diversity, all these species-specific indices should be viewed as heuristics; how well prioritising based on these indices do at such maximisation remains an open empirical question.

\section{Acknowledgements}
We thank the NZ Marsden Fund (23-UOC-003) and the NSERC Canada Discovery Program for research support, and Alex Pyron for help sourcing Crocodile data.
We also thank the three anonymous reviewers for their helpful comments that improved both the technical presentation and the framing of this work.

%\section{Supplementary Material}
%Data and sources for the Crocodylia analyses available from the Dryad Digital Repository:
%https://doi.org/10.5061/dryad.v9s4mw79j

\end{document}